\documentclass[twoside]{article}
\title{\sc Polymatroids and polyquantoids\footnotemark}
\author{{\bf Franti\v{s}ek Mat\'{u}\v{s}}\\[2ex]
         \small Institute of Information Theory and Automation \\
         \small Academy of Sciences of the Czech Republic\\
         \small {\tt matus@utia.cas.cz}    }
\date{}
\newfont{\m}{cmr8}\newfont{\ms}{cmsl8}
\usepackage[T1]{fontenc}\usepackage{amsmath,amsfonts,amssymb,amsthm,xspace}
    \theoremstyle{plain}  \newtheorem{lemma}{Lemma}%[section]
                          \newtheorem{theorem}{Theorem}
                          \newtheorem{corollary}{Corollary}
                          
    \theoremstyle{definition}  
    \theoremstyle{remark}      \newtheorem{remark}{Remark}
   \DeclareMathAlphabet{\mathsfsl}{OT1}{cmss}{m}{sl}
   \DeclareMathAlphabet{\mathbfsl}{OT1}{cmr}{bx}{it}
      
      \renewcommand{\geq}{\geqslant}\renewcommand{\leq}{\leqslant}
      \newcommand{\pmn}{\emptyset}\newcommand{\pdm}{\subseteq}\newcommand{\sm}{\setminus}
      
    \newcommand{\hatt}{^{\scriptscriptstyle\wedge}}\newcommand{\chee}{^{\scriptscriptstyle\vee}}
    \newcommand{\nnn}{{\mathsfsl0}}
        \newcommand{\vre}{\phi^*}\newcommand{\zje}{\mbox{$\bigcup$}}
        \newcommand{\exd}{^{\scriptscriptstyle \#}}

\begin{document}
\maketitle
    \footnotetext{This work was supported by Grant Agency the Czech Republic under Grant 201/08/0539.}
    \begin{abstract}
       When studying entropy functions of multivariate probability distributions,
       polymatroids and matroids emerge. Entropy functions of pure multiparty quantum
       states give rise to analogous notions, called here polyquantoids and quantoids.
       Polymatroids and polyquantoids are related via linear mappings and duality.
       Quantum secret sharing schemes that are ideal are described by selfdual matroids.
       Expansions of integer polyquantoids to quantoids are studied and linked to that
       of polymatroids.
    \end{abstract}
\begin{picture}(0,0)
    \put(0,290){\makebox(0,0)[l]{\small  \emph{Proceedings of WUPES 2012}, Sept.\ 12--15, Mari\'ansk\'e L\'azn\v{e},  
                                        Czech Republic, 126--136.}}
\end{picture}

\thispagestyle{empty}

%111111111111111111111111111111111111111111111111111111111111111111111111111111111111111111111111
\section{Introduction\label{S:intro}}

 A polymatroid $(N,h)$ consists of a finite ground set $N$ and rank
 function $h$ on the subsets of $N$ that is normalized $h(\pmn)=0$,
 nondecreasing  \mbox{$h(I)\leq h(J)$}, $I\pdm J$, and submodular
 $h(I)+h(J)\geq h(I\cup J)+h(I\cap J)$, $I,J\pdm N$. A polymatroid is
 entropic if there exists a probability measure $P$ on a finite set
 $\prod_{i\in N}\: X_i$ such that $h(I)$ equals Shannon entropy of the
 marginal of $P$ to $\prod_{i\in I}\: X_i$, for all \mbox{$I\pdm N$}.
 This means that $h$ equals the entropy function of $P$. These functions
 always induce polymatroids.

 In this work, a \emph{polyquantoid} is introduced as a pair $(N,e)$
 with a rank function $e$ on the subsets of $N$ that is normalized,
 comple\-mentary $e(I)=e(N\sm I)$, $I\pdm N$, and submodular. A polyquantoid
 is entropic if there exists a quantum state $\rho$ on a complex Hilbert space
 $\bigotimes_{i\in N}\: H_i$ of finite dimension such that $e(I)$ equals von
 Neumann entropy of the reduction of $\rho$ to $\bigotimes_{i\in I}\: H_i$,
 for all $I\pdm N$. This means that $e$ equals the entropy function on $\rho$.
 These functions always induce polyquantoids, by properties of von Neumann
 entropy.

 A polymatroid/polyquantoid is integer if all values of its rank function
 are integer numbers. An integer polymatroid whose values on singletons
 equal zero or one is called matroid. Let \emph{quantoid} be defined as
 an integer polyquantoid with this property.

 This contribution studies interplay between polymatroids, polyquantoids,
 matroids, quantoids and secret sharing schemes, both classical and quantum.
 In Section~\ref{S:dual}, duality of set functions is worked out.
 Section~\ref{S:relation} introduces mutually inverse linear mappings
 that provide a one-to-one correspondence between tight selfdual
 polymatroids and polyquantoids, see Theorem~\ref{T:tamaspat}. This
 correspondence can serve as a tool for comparing problems on
 classical and quantum entropy functions.

 In Section~\ref{S:isss}, secret sharing schemes are lifted to the level
 of poly\-matroids/poly\-quan\-toids.  Theorem~\ref{T:polymatroidsss} recalls
 that the ideal sharing in polymatroids is governed by matroids. This result
 is translated to polyquantoids in Theorem~\ref{T:polyquantoidsss} that
 describes the ideal quantum sharing via those quantoids that correspond
 to tight selfdual matroids.

 Section~\ref{S:expan} departs from the notion of expansions of integer
 polymatroids to matroids. An analogous construction for integer polyquantoids
 is introduced to provide expansions of polyquantoids to quantoids, see
 Theorem~\ref{T:expan.seldual.polyquant}. Thus, the quantoids play a role
 of matroids in quantum settings. In Section~\ref{S:disc}, remarks and
 discussion of related material and literature are collected.

%222222222222222222222222222222222222222222222222222222222222222222222222222222222222222222222222
\section{Duality\label{S:dual}}

 For set functions $h$ with a ground set $N$, the following definition
 \[
     h'(I)\triangleq h(N\sm I) + h(\pmn) -h(N)
            +\sum_{i\in I}\big[h(i)-h(\pmn)+h(N)-h(N\sm i)\big]\,,\quad I\pdm N\,,
 \]
 gives rise to a duality mapping $h\mapsto h'$. A function $h$ is \emph{selfdual}
 if $h'=h$. The functions that are comple\-mentary, as in polyquantoids, are selfdual.

 Let us say that a set function $h$ is \emph{tight} if $h(N\sm i)=h(N)$
 for all $i\in N$. If $h$ is normalized and tight then the definition
 of duality simplifies to
 \[
    h'(I)=h(N\sm I)-h(N) +\sum_{i\in I} h(i)\,,\quad I\pdm N\,.
 \]

\begin{lemma}\label{L:dual}
    For any function $h$ on the subsets of $N$,\\
    \indent (i) $h'(\pmn)=h(\pmn)$,\\
    \indent (ii) $h'(i)=h(i)$ for $i\in N$,\\
    \indent (iii) $h'(N)-h'(N\sm i)=h(N)-h(N\sm i)$ for $i\in N$,\\
    \indent (iv) $h''=h$,\\
    \indent (v) $h$ is submodular if and only if $h'$ is so,\\
    \indent (vi) if $h$ is normalized, submodular and $h(N)\geq h(N\sm i)$,
                    $i\in N$, then $h'$ is nondecreasing.
\end{lemma}

\begin{proof}
 First two assertions follow directly from the definition. For $K\pdm J$ the equality
 \[
    h'(J)-h'(K)=h(N\sm J)-h(N\sm K)
                    +\sum_{i\in J\sm K}\big[h(i)-h(\pmn)+h(N)-h(N\sm i)\big]
 \]
 implies \emph{(iii)}. Choosing $K=N\sm I$ and $J=N$, it rewrites to
 \[
    h(I)=h'(N\sm I) + h(\pmn) -h'(N)
            +\sum_{i\in I}[h(i)-h(\pmn)+h(N)-h(N\sm i)\big]\,,\qquad I\pdm N\,.
 \]
 By \emph{(i)}, \emph{(ii)} and \emph{(iii)}, the right-hand side equals $h''(I)$
 which proves \emph{(iv)}. If $h$ is submodular then $I\mapsto h(N\sm I)$ is so whence
 $h'$ is submodular. Then, the equivalence \emph{(v)} holds by~\emph{(iv)}. If $h$
 is normalized and $h(N)\geq h(N\sm i)$, $i\in N$, then for $J\supseteq K$
 \[
    h'(J)-h'(K)\geq h(N\sm J)-h(N\sm K)+\sum_{i\in J\sm K} h(i)\,.
 \]
 If $h$ is also submodular then the right-hand side is nonnegative whence \emph{(vi)} holds.
\end{proof}

\begin{corollary}\label{C:dual}
     The duality mapping restricts to an involution on the (tight) polymatroids.
\end{corollary}

%3333333333333333333333333333333333333333333333333333333333333333333333333333333333333333333333
\section{Tight selfdual polymatroids and polyquantoids\label{S:relation}}

 Let $h$ and $e$ be set functions with the ground set $N$. The linear
 mappings $e\mapsto e\hatt$ and $h\mapsto h\chee$ introduced here by
 \[
     e\hatt(I)\triangleq e(I)+\sum_{i\in I} e(i)\quad\text{and}\quad
     h\chee(I)\triangleq h(I) -\tfrac12 \sum_{i\in I} h(i)\,,\quad I\pdm N\,,
 \]
 are mutually inverse, $(e\hatt)\chee=e$ and $(h\chee)\hatt=h$. They
 provide a natural link between the polymatroids and polyquantoids.

\begin{theorem}\label{T:tamaspat}
     The mappings $e\mapsto e\hatt$ and $h\mapsto h\chee$ restrict to mutually
     inverse bijections between the polyquantoids and the tight selfdual polymatroids.
\end{theorem}

\begin{proof}
 Let $(N,e)$ be a polyquantoid. Since $e$ is normalized $e\hatt(\pmn)=0$.
 The submodularity of $e$ is equivalent to that of $e\hatt$, and
 implies $e(N\sm I)\leq e(N\sm J)+\sum_{i\in J\sm I} e(i)$ for $I\pdm J\pdm N$.
 By complementarity, $e(I)\leq e(J)+\sum_{i\in J\sm I} e(i)$, and thus
 $e\hatt(I)\leq e\hatt(J)$. Therefore, $(N,e\hatt)$ is a polymatroid.
 Since $e$ is normalized and complementary
 \[
     e\hatt(N)=\sum_{j\in N} e(j)=e(N\sm i)+\sum_{j\in N\sm i} e(j)=e\hatt(N\sm i)
     \,,\quad i\in N\,,
 \]
 thus $e\hatt$ is tight. For $I\pdm N$ it follows that
 \[\begin{split}
    (e\hatt)'(I)&=e\hatt(N\sm I)- e\hatt(N) + \sum_{i\in I} e\hatt(i)\\
                &=\Big[e(N\sm I)+\sum_{i\in N\sm I}e(i)\Big]
                    -\sum_{i\in N}e(i) + 2\sum_{i\in I}e(i)=e\hatt(I)\,,
 \end{split}\]
 thus $e\hatt$ is selfdual.

 Let $(N,h)$ be a tight selfdual polymatroid. Since $h$ is normalized $h\chee(\pmn)=0$.
 Since $h$ is tight and selfdual $h(I)=h(N\sm I)-h(N) +\sum_{i\in I} h(i)$, $I\pdm N$.
 Then, $h(N)$ is equal to $\tfrac12 \sum_{i\in N} h(i)$. It follows that
 \[
     h\chee(N\sm I)=\Big[h(I)-h(N)+\sum_{i\in N\sm I} h(i)\Big] -\tfrac12 \sum_{i\in N\sm I} h(i)
                    =h\chee(I)\,,\quad I\pdm N\,,
 \]
 thus $h\chee$ is complementary. The submodularity of $h$ implies that of~$h\chee$.
 Therefore, $(N,h\chee)$ is a polyquantoid.
\end{proof}

\begin{remark}
 The above proof provides also arguments for the assertion that the mappings $e\mapsto e\hatt$
 and $h\mapsto h\chee$ restrict to mutually inverse bijections between the class of normalized
 complementary functions and the class of normalized tight selfdual functions, dropping
 submodularity in Theorem~\ref{T:tamaspat}.
\end{remark}

\begin{corollary}\label{C:integer}
     The mappings $e\mapsto e\hatt$ and $h\mapsto h\chee$ induce mutually inverse bijections
     between the integer polyquantoids and the integer tight selfdual polymatroids whose values
     on all singletons are even.
\end{corollary}

\begin{corollary}\label{C:quantoid}
     The mappings $e\mapsto e\hatt$ and $h\mapsto h\chee$ induce mutually inverse bijections
     between the quantoids and the integer tight selfdual polymatroids whose values
     on all singletons equal zero or two.
\end{corollary}

%4444444444444444444444444444444444444444444444444444444444444444444444444444444444444444444444444
\section{Ideal secret sharing\label{S:isss}}

 Given a polymatroid $(N,h)$, an element $\nnn$ of $N$ is \emph{perfect}
 if $h(\nnn\cup I)-h(I)$ equals $h(\nnn)$ or zero, for all $I\pdm N\sm\nnn$.
 In the latter case, $I$ is \emph{authorized} for $\nnn$. By submodularity,
 \[
    h(\nnn\cup I)-h(I)\geq h(\nnn\cup J)-h(J)\,,\qquad I\pdm J\pdm N\sm \nnn\,.
 \]
 Hence, $h(\nnn\cup I)-h(I)=0$ implies $0\geq h(\nnn\cup J)-h(J)$, and
 $h(\nnn\cup J)-h(J)=h(\nnn)$ implies $h(\nnn\cup I)-h(I)\geq h(\nnn)$.
 The two inequalities are tight as $h$ is a polymatroid. Thus, the family
 of authorized sets for $\nnn$ is closed to supersets and the family of
 sets $I\pdm N\sm\nnn$ with $h(\nnn\cup I)-h(I)$ equal to $h(\nnn)$ is
 closed to subsets. This is referred to as heredity. If $\nnn$ is perfect
 and $h(\nnn)>0$ then the two families are disjoint and cover all subsets
 of $N\sm\nnn$, which is referred to as dichotomy.

 In a polymatroid $(N,h)$ with a perfect element $\nnn\in N$, an element
 $i\in N\sm\nnn$ is \emph{essential} for $\nnn$ if it belongs to some set
 $I$ that is authorized for $\nnn$ and $h(\nnn\cup I\sm i)-h(I\sm i)=h(\nnn)$.
 As a consequence,
 \[
     h(i)\geq h(I)-h(I\sm i)=h(\nnn\cup I)-h(I\sm i)
      \geq h(\nnn\cup I\sm i)-h(I\sm i)=h(\nnn)\,,
 \]
 since $h$ is submodular and nondecreasing. A perfect element $\nnn$ in
 a polymatroid $(N,h)$ is \emph{ideal} if each $i\in N\sm\nnn$ is essential
 for $\nnn$ and $h(i)=h(\nnn)$.

 For example, in any matroid $(N,r)$ each element is perfect. Given $\nnn\in N$,
 a set $I\pdm N\sm\nnn$ is authorized for $\nnn$ if and only if a circuit contained
 in $\nnn\cup I$ contains $\nnn$. If $r(\nnn)=0$, thus $\nnn$ is a loop, then all
 $i\in N\sm\nnn$ are essential for $\nnn$. Hence, $\nnn$ is ideal if only if $r(N)=0$.
 Otherwise, when $r(\nnn)=1$, $i$ is essential for $\nnn$ if and only if there exists
 a circuit of the matroid containing $\nnn$ and~$i$. Therefore, $\nnn$ is ideal if only
 if the matroid is connected. Each element of any connected matroid is ideal.

 When restricting to the entropic polymatroids, the above notions correspond to the
 information-theoretical secret sharing schemes.

 The following assertion claims that existence of an ideal element
 implies matroidal structure. It follows from an existing result,
 see~Section~\ref{S:disc}, but a self-contained proof is presented for
 convenience.

\begin{theorem}\label{T:polymatroidsss}
    If a polymatroid $(N,h)$ has an ideal element then there exists
    a matroid $(N,r)$ and $t>0$ such that $h=t\,r$.
\end{theorem}

%
% the result can be likely generalized to semigraphoids with functional dependence
% (plus something that captures h(\nnn)=h(i))
% which would cover the theorem of Blackley and Kabatianski
%

\begin{proof}
 Let $\nnn\in N$ be an ideal element of the polymatroid. If $h(\nnn)=0$ then
 $h(i)=0$ for all $i\in N$ whence $(N,h)$ is a matroid and the assertion
 holds with any $t>0$. Let $h(\nnn)>0$.

 The idea is to prove that `if $L\pdm N$ is nonempty then there exists
 $\ell\in L$ such that $h(L)-h(L\sm\ell)$ equals $h(\nnn)$ or zero'. This
 implication and an induction argument on the cardinality of $L$ show
 that all values of $h$ are multiples of $h(\nnn)$. As a consequence,
 $h$ equals a matroid rank function multiplied by $t=h(\nnn)>0$.

 If $L\pdm N$ contains $\nnn$ the implication holds with $\ell=\nnn$
 because $\nnn$ is perfect.

 If $L\pdm N\sm\nnn$ is authorized, $h(\nnn\cup L)=h(L)$, then
 $h(\nnn\cup I)=h(\nnn)+h(I)$ for some $I\pdm L$, e.g.\ $I=\pmn$.
 Such a set $I$ is chosen to be inclusion maximal. By dichotomy,
 $I \varsubsetneq L$. Let $\ell\in L\sm I$. Since $I$ is maximal
 and $\nnn$ perfect, $\ell\cup I$ is authorized,
 $h(\nnn\cup \ell\cup I)=h(\ell\cup I)$. This
 and submodularity imply
 \[\begin{split}
    h(\nnn\cup L\sm \ell)+h(\nnn\cup \ell\cup I)&\geq
        h(\nnn\cup L)+h(\nnn\cup I)=h(\nnn\cup L)+h(\nnn)+h(I)\,,\\
    h(\ell)+h(I)&\geq h(\ell\cup I)=h(\nnn\cup \ell\cup I)\,.
 \end{split}\]
 As $\nnn$ is ideal, $h(\nnn)=h(\ell)$, and it follows by adding that
 $h(\nnn\cup L\sm \ell)\geq h(\nnn\cup L)$. Thus,
 $h(\nnn\cup L\sm \ell)=h(\nnn\cup L)=h(L)$ because
 $h$ is nondecreasing and $L$ authorized. This implies that
 $h(L)-h(L\sm \ell)$ equals $h(\nnn\cup L\sm \ell)-h(L\sm \ell)$
 which is zero or $h(\nnn)$ by perfectness of~$\nnn$.
 Hence, the implication holds for every $L$ authorized.

 By dichotomy, it remains to consider a nonempty subset $L$ of $N\sm\nnn$
 such that $h(\nnn\cup L)$ equals $h(\nnn)+h(L)$. Since $\nnn$ is ideal,
 any $\ell\in N\sm\nnn$ is essential for $\nnn$. Taking some $\ell\in L$
 there exists an authorized set $K$, $h(\nnn\cup K)=h(K)$, such that $\ell\in K$
 and $h(\nnn\cup K\sm\ell)$ equals $h(\nnn)+h(K\sm\ell)$. Such a set $K$
 is chosen to obtain the cardinality of $K\sm L$ minimal. By dichotomy, $K$ is not
 contained in $L$. For every $k\in K\sm L$ the minimality implies that the set
 $L\cup K\sm k$, containing the chosen $\ell$, is not authorized. In turn, since $h$ is
 submodular, $L\cup K$ authorized and $h$ nondecreasing
 \[
    h(k)+h(L\cup K\sm k)\geq h(L\cup K)=h(\nnn\cup L\cup K)
        \geq h(\nnn\cup L\cup K\sm k)=h(\nnn)+h(L\cup K\sm k)\,.
 \]
 The above two inequalities are tight because $h(\nnn)=h(k)$, using that $\nnn$
 is ideal. Therefore, $h(L\cup K)=h(k)+h(L\cup K\sm k)$ for $k\in K\sm L$. By
 induction,
 \[
    h(I\cup (K\sm L))=h(I)+\sum_{k\in K\sm L}\,h(k)\,, \qquad I\pdm L\,.
 \]
 This implies that $h(L)-h(L\sm\ell)$ equals $h(L\cup K)-h((L\cup K)\sm\ell)$.
 The previous part of the proof is applied to the authorized set $K$ in the
 role of $L$ and the non-authorized set $K\sm\ell$ in the role of $I$ to
 conclude that $h(\nnn\cup K\sm \ell)=h(\nnn\cup K)$. This implies that
 $h(\nnn\cup(L\cup K)\sm\ell)$ equals $h(\nnn\cup L \cup K)$ which coincides
 with $h(L\cup K)$ because $L\cup K$ is authorized. Hence, $h(L)-h(L\sm\ell)$
 equals $h(\nnn\cup (L \cup K)\sm\ell)-h((L\cup K)\sm\ell)$ which is zero or
 $h(\nnn)$ by perfectness of $\nnn$. Thus, the implication holds for all
 nonempty $L\pdm N\sm\nnn$.
\end{proof}

 Given a polyquantoid $(N,e)$, an element $\nnn$ of $N$ is \emph{perfect} if
 $e(\nnn\cup I)-e(I)$ equals $e(\nnn)$ or $-e(\nnn)$, for all $I\pdm N\sm\nnn$.
 In the latter case, $I$ is \emph{authorized} for $\nnn$. The definition of
 perfectness does not change when requiring that $e\hatt(\nnn\cup I)-e\hatt(I)$
 equals $e\hatt(\nnn)$ or zero. Thus, $\nnn$ is perfect in $(N,e)$ if and only if
 it is perfect in the polymatroid $(N,e\hatt)$. Therefore, supersets of authorized
 sets are authorized and the equality $e(\nnn\cup I)-e(I)=e(\nnn)$ with
 $I\pdm N\sm\nnn$ is inherited by the subsets of $I$. The dichotomy takes
 place whenever $e(\nnn)>0$.

 In a polyquantoid $(N,e)$ with a perfect element $\nnn\in N$, an element
 $i\in N\sm\nnn$ is \emph{essential} for $\nnn$ if there exists a set $I$
 which authorized for $\nnn$, contains $i$ and $e(\nnn\cup I\sm i)-e(I\sm i)=e(\nnn)$.
 This is equivalent to saying that $i\in N\sm\nnn$ is \emph{essential} for $\nnn$ in the
 polymatroid $(N,e\hatt)$. Hence, $e(i)\geq e(\nnn)$ once $i$ is essential for $\nnn$ in $(N,e)$.
 A perfect element $\nnn$ in a polyquantoid $(N,e)$ is \emph{ideal} if each
 $i\in N\sm\nnn$ is essential for $\nnn$ and $e(i)=e(\nnn)$.

\begin{theorem}\label{T:polyquantoidsss}
    If a polyquantoid $(N,e)$ has an ideal element then there exists
    a tight selfdual matroid $(N,r)$ and $t>0$ such that $e=t\, r\chee$.
\end{theorem}

\begin{proof}
 If $\nnn\in N$ is ideal in the polyquantoid then $\nnn$ is ideal in $(N,e\hatt)$
 which is a tight selfdual polymatroid by Theorem~\ref{T:tamaspat}.
 Theorem~\ref{T:polymatroidsss} implies that $e\hatt=t\,r$ for $t>0$
 and a matroid rank function $r$. Hence, $r$ is tight, selfdual,
 and $e=(e\hatt)\chee=(t\,r)\chee=t\,r\chee$.
\end{proof}

\noindent
 As a consequence, if $\nnn$ is an ideal element of a polyquantoid then
 $I\pdm N\sm\nnn$ is authorized for $\nnn$ if and only if $\nnn\in C\pdm\nnn\cup I$
 for some circuit $C$ of the tight selfdual matroid that is assigned to the
 polyquantoid in Theorem~\ref{T:polyquantoidsss}.

%55555555555555555555555555555555555555555555555555555555555555555555555555555555555555555555
\section{Expansions\label{S:expan}}

 A set function $h$ with a ground set $N$ \emph{expands} to a set function $h\exd$
 with  a ground set $N\exd$ if there exists a mapping $\phi$ on $N$ ranging in the
 family of subsets of $N\exd$ such that $h\exd(\bigcup_{i\in I}\phi(i))$ equals $h(I)$
 for all $I\pdm N$.

 Each integer polymatroid $(N,h)$ can be expanded to a matroid as follows.
 Let $\phi$ map $i\in N$ to a set $\phi(i)$ of cardinality $h(i)$ such
 that these sets are pairwise disjoint. Writing
 $\phi(I)=\bigcup_{i\in I}\,\phi(i)$, $I\pdm N$, the construction
 \[
       h_\phi\colon K\mapsto \min_{J\pdm N}\:[\, h(J)+ |K\sm \phi(J)|\,]
            \,,\qquad K\pdm \phi(N)\,,
 \]
 defines a matroid $(\phi(N),h_\phi)$ called a \emph{free expansion} of
 $(N,h)$. The value $h_\phi(K)$
 depends on $K$ only through the cardinalities of the sets $\phi(i)\cap K$,
 $i\in N$. The minimization can be equivalently over the sets that satisfy
 \[
    \{i\in N\colon \phi(i)\cap K\neq\pmn\}\supseteq J\supseteq
    \{i\in N\colon \pmn\neq\phi(i)\pdm K\}
 \]
 since $h$ is nondecreasing and submodular. Such sets $J$ are termed to
 be \emph{adapted} to $K$. Hence, $h_\phi(\phi(I))$ equals $h(I)$ for all
 $I\pdm N$, using that $\{i\in I\colon\phi(i)\neq\pmn\}$ is the unique
 adapted set to $\phi(I)$, and thus $h$ expands to~$h_\phi$.

 For any integer polyquantoid $(N,e)$, an analogous construction is introduced
 as follows. Let $\psi$ map $i\in N$ to a set $\psi(i)$ of cardinality $e(i)$
 such that these sets are pairwise disjoint, $\psi(I)=\bigcup_{i\in I}\,\psi(i)$,
 $I\pdm N$, and
 \[
       e_\psi\colon K\mapsto \min_{J\pdm N}\:[\, e(J)+ |K\vartriangle\psi(J)|\,]
            \,,\qquad K\pdm \psi(N)\,.
 \]
 Let $(\psi(N),e_\psi)$ be called a \emph{free expansion} of $(N,e)$.
 The minimization can be equivalently over the sets adapted to $K$,
 using that $e$ is normalized, complementary and submodular. Therefore,
 $e_\psi(\psi(I))=e(I)$, $I\pdm N$, thus $e$ expands to~$e_\psi$ indeed.

 The following assertion shows that from the viewpoint of expansions,
 quantoids are for polyquantoids what matroids are to polymatroids.

\begin{theorem}\label{T:expan.seldual.polyquant}
    Any free expansion of an integer polyquantoid is a quantoid.
\end{theorem}

\begin{proof}
 Let $(N,e)$ be an integer polyquantoid and $\psi$ a mapping as above.
 By definition, $e_\psi(K)=|K|$ if $K\pdm\psi(i)$ for some $i\in N$. In
 particular, $e_\psi$ is normalized and its values on singletons equal one.

 For $J\pdm N$ adapted to $K\pdm\psi(N)$, the set
 $J'=\{i\in N\sm J\colon\psi(i)\neq\pmn\}$ is adapted
 to $\psi(N)\sm K$ and
 \[
    e(J)+|K\vartriangle\psi(J)|
            =e(J')+\big|(\psi(N)\sm K)\vartriangle(\psi(J'))\big|
 \]
 using $e(J)=e(N\sm J)=e(J')$. Moreover, $J\mapsto J'$ is a bijection
 between the families of those sets that are adapted to $K$, resp.~to
 $\psi(N)\sm K$. It follows by minimization that $e_\psi(K)$ equals
 $e_\psi(\psi(N)\sm K)$, thus $e_\psi$ is complementary.

 To prove that $e_\psi$ is submodular, let $K,L\pdm\psi(N)$ and
 \[
       e_\psi(K)=e(I)+|K\vartriangle \psi(I)|\quad\text{and}\quad
       e_\psi(L)=e(J)+|L\vartriangle \psi(J)|\,
 \]
 where $I$ is adapted to $K$ and $J$ is adapted to $L$.
 As $e(I)+e(J)\geq e(I\cup J)+e(I\cap J)$ and
 \[
    |K\vartriangle\psi(I)|+|L\vartriangle\psi(J)|=
        |(K\cup L)\vartriangle\psi(I\cup J) |
            +|(K\cap L)\vartriangle\psi(I\cap J)|
 \]
 the submodularity of $e_\psi$ follows.
\end{proof}

 In the remaining part of this section, expansions of polymatroids
 and polyquantoids are compared by means of the mappings
 $e\mapsto e\hatt$ and $h\mapsto h\chee$.

 Let $(N,h)$ be an integer polymatroid with $h(i)$ even for all $i\in N$
 and $(\phi(N),h_\phi)$ its free expansion. Then, each set $\phi(i)$ can be
 partitioned into two-element blocks $m=\{k,\ell\}$ having $k,\ell\in\phi(i)$
 different. Let $\vre(i)$ denote the set of all blocks in such a partition,
 $\vre(I)=\bigcup_{i\in I}\,\vre(i)$, $I\pdm N$, and
 \[
       h_{\vre}(M)\triangleq h_\phi(\zje M)
        =\min_{J\pdm N}\:[\, h(J)+ |\zje M\sm \phi(J)|\,]
            \,,\qquad M\pdm \vre(N)\,.
 \]
 This defines a polymatroid $(\vre(N),h_{\vre})$ called here \emph{$2$-factor}
 of $(\phi(N),h_{\phi})$. By definitions, $(N,h)$ expands to $(\vre(N),h_{\vre})$
 which in turn expands to $(\phi(N),h_{\phi})$.

 The following assertion indicates a correspondence between the free expansions
 of polymatroids and polyquantoids.

\begin{lemma}\label{L:polyq.to.polym}
   If $(N,e)$ is an integer polyquantoid, $h=e\hatt$, $(\phi(N),h_\phi)$
   a free expansion of $(N,h)$ and $(\vre(N),h_{\vre})$ its $2$-factor
   then $(\vre(N),(h_{\vre})\chee)$ is a free expansion of $(N,e)$.
\end{lemma}

\begin{proof}
 For $M\pdm\vre(N)$
 \[
    (h_{\vre})\chee(M)=h_{\vre}(M)-\tfrac12\sum_{m\in M}\,h_{\vre}(\{m\})
                      =h_\phi(\zje M)-|M|
 \]
 using that $h_{\vre}(\{m\})=h_\phi(m)=2$. Since $e(j)=h(j)/2=|\vre(j)|$
 for $j\in N$, if $J\pdm N$ then $h(J)=e(J)+|\vre(J)|$. Then, by the
 definition of polymatroid expansions,
 \[
    (h_{\vre})\chee(M)= \min_{J\pdm N}\:
                \big[\,e(J)+|\vre(J)|+|\zje M\sm \phi(J)|\,\big]-|M|
 \]
 Here, $|\zje M\sm \phi(J)|=2|M\sm\vre(J)|$. Since $|\vre(J)|+|M\sm\vre(J)|-|M|$
 equals $|\vre(J)\sm M|$ it follows from definition of polyquantoid expansions
 that $(h_{\vre})\chee$ coincides with $e_{\vre}$.
\end{proof}

 In the above lemma, the integer polymatroid $h=e\hatt$ is tight and selfdual,
 by Theorem~\ref{T:tamaspat}. The following two lemmas imply that the expansion
 $h_\phi$ and its $2$-factor $h_{\vre}$ have the same properties. Hence,
 Theorem~\ref{T:expan.seldual.polyquant} can be proved alternatively by combining
 Theorem~\ref{T:tamaspat} with Lemmas~\ref{L:polyq.to.polym}, \ref{L:exp.polymat}
 and \ref{L:coa.polymat}. This argument is more involved but illustrates the
 interplay between the two kinds of expansions.

\begin{lemma}\label{L:exp.polymat}
   If an integer polymatroid is tight and selfdual then so are its free expansions.
\end{lemma}

\begin{proof}
 Let $(N,h)$ be an integer polymatroid and $\phi$ a mapping with $|\phi(i)|=h(i)$,
 $i\in N$, as above. For $k\in\phi(N)$ there exists unique $i\in N$ such that
 $k\in\phi(i)$. Assuming that $h$ is tight $h_\phi(\phi(N\sm i))=h(N\sm i)=h(N)=h_\phi(\phi(N))$.
 This implies  $h_\phi(\phi(N)\sm k)=h_\phi(\phi(N))$ whence $h_\phi$ is tight.

 By definition, $h_\phi(K)=|K|$ if $K\pdm\phi(i)$ for some $i\in N$. Hence,
 assuming that $h$ is tight and selfdual, for a set $J\pdm N$ adapted to $K\pdm\phi(N)$,
 \[\begin{split}
    h(J)+|K\sm\phi(J)|&=h(N\sm J)-h(N)+|\phi(J)|+|K\sm\phi(J)|\\
            &=h(N\sm J)-h_\phi(\phi(N))+|K|+\big|(\phi(N)\sm K)\sm(\phi(N\sm J))\big|\,.
 \end{split}\]
 Minimizing over the adapted sets, it follows that
 $h_\phi(K)\geq h_\phi(\phi(N)\sm K)-h_\phi(\phi(N))+|K|$. Since $J$ is adapted
 to $K$ if and only if $J'=\{i\in N\sm J\colon\phi(i)\neq\pmn\}$ is adapted
 to $\phi(N)\sm K$ this inequality is tight. Thus, $h_\phi$ is selfdual.
\end{proof}

\begin{lemma}\label{L:coa.polymat}
   If an integer polymatroid is tight, selfdual and takes even values
   on all singletons then all $2$-factors of its free expansions are
   tight and selfdual.
\end{lemma}

\begin{proof}
 Let $(N,h)$ satisfy the assumptions. Keeping the notation of the proof
 of Lemma~\ref{L:exp.polymat}, for $m\in\vre(N)$ there exists unique $i\in N$
 such that $m\pdm\phi(i)$. Since $h$ is tight $h_\phi(\phi(N\sm i))$ equals
 $h_\phi(\phi(N))$. Hence,
 \[
    h_{\vre}(\vre(N)\sm\{m\})=h_\phi(\phi(N)\sm m)\geq h_\phi(\phi(N\sm i))
                             =h_\phi(\phi(N))=h_{\vre}(\vre(N))
 \]
 In turn, $h_{\vre}$ is tight.

 By Lemma~\ref{L:exp.polymat}, $(\phi(N),h_\phi)$ is selfdual. Hence, for $M\pdm\vre(N)$
 \[\begin{split}
    (h_{\vre})'(M)&=h_{\vre}(\vre(N)\sm M)-h_{\vre}(\vre(N)) +\sum_{m\in M} h_{\vre}(\{m\})\\
                  &=h_{\phi}(\phi(N)\sm \zje M)-h_{\phi}(\phi(N)) +\sum_{k\in\mbox{$\scriptstyle\bigcup$} M} h_{\phi}(k)
                  =h_{\phi}(\zje M)=h_{\vre}(M)
 \end{split}\]
 using that $ h_{\vre}(\{m\})=2= h_{\phi}(k)+h_{\phi}(\ell)$ where $m=\{k,\ell\}$.
\end{proof}

%6666666666666666666666666666666666666666666666666666666666666666666666666666666666666666666
\section{Discussion\label{S:disc}}

 The polymatroids \cite{Lov,Fujibook,Nar} have been studied for decades
 and history of the matroid theory \cite{Oxley} is even longer. The duality
 defined in Section~\ref{S:dual} is in general different from known ones, as
 those in~\cite{Nar,Oxley,Whittle}, since it conserves values on singletons,
 see Lemma~\ref{L:dual}\emph{(ii)}. For matroids without loops and coloops,
 the duality coincides with the usual one~\cite[2.1.9]{Oxley}. Functions called
 above selfdual are in literature also termed identically self\-dual.
 Tightness is a notion suitable for this work but not used elsewhere.
 A matroid is tight if and only if it has no coloop.

 The problem which polymatroid is entropic is of interest for
 information-theoretical approaches to networks and cryptography, and
 beyond, for references see e.g.\ \cite{ZhY.ineq,M.infinf,M.twocon}. Its
 quantum version, asking which polyquantoid is entropic, has also attracted
 considerable attention \cite{Pippenger,LindenWinter:new,Ca.Li.Wi}.

 Ideal secret sharing schemes were investigated first in a combinatorial
 setting \cite{BriDa}. Theorem~\ref{T:polymatroidsss} is a consequence
 of \cite[Theorem 2]{BlaKa}, building on \cite[Theorem~1]{BriDa}. The presented
 proof is based on the approach of~\cite{BlaKa}. Quantum secret sharing schemes
 go back to \cite{CleveGottLo,Hi.Bu.Be,Gott.qsss}. Ideal sharing and matroids
 were discussed recently in \cite{SaRu,Sa.sympl}. Theorem~\ref{T:polyquantoidsss}
 solves a question related to \cite[Fig.~2]{SaRu}. It implies that the access
 structure of any ideal quantum secret sharing scheme must be generated
 by circuits of a tight selfdual matroid.

 Free expansions were proposed independently by several researchers, see
 \cite{Helg,McDiarmid.Rado,Nguyen}. If an entropic integer polymatroid
 expands to a matroid then the latter is the limit of entropic polymatroids
 \cite[Theorem~4]{M.twocon}. The quantum analogue of this assertion is open.

\section*{Acknowledgement\label{S:ack}}

 The author wants to express sincere thanks to Mary Beth Ruskai
 and Andreas Winter for numerous inspiring discussions and
 kind hospitality during visits. Discussions with Oriol Farr\`as
 Ventura on quantum secret sharing are acknowledged.

\end{document}